\documentclass[11pt]{article}
\usepackage[utf8]{inputenc}
\usepackage{amsmath, amssymb}
\usepackage{amsfonts}
\usepackage{amsthm}
\usepackage{graphicx}
\usepackage{csquotes}
\usepackage{xspace}
\usepackage{fullpage}
\usepackage{algpseudocode}
\usepackage{url}
\usepackage{hhline}
\usepackage{amscd}

\usepackage{tikz,pgfplots}
\usetikzlibrary{shapes.geometric}
\usepackage{wrapfig}
\usepackage{csquotes}
\usepackage{thmtools}
\usepackage{thm-restate}
\usepackage{tikz-cd}

\graphicspath{{./images/}}
\newtheorem{exmp}{Example}[section]
\newtheorem{theorem}{Theorem}[section]
\newtheorem{definition}[theorem]{Definition}
\newtheorem{lemma}[theorem]{Lemma}
\newtheorem{corollary}[theorem]{Corollary}
\newtheorem{proposition}[theorem]{Proposition}
\newtheorem{conjecture}[theorem]{Conjecture}

\renewcommand\appendix{\par
	\setcounter{section}{0}
	\setcounter{subsection}{0}
	\setcounter{figure}{0}
	\setcounter{table}{0}
	\renewcommand\thesection{\Alph{section}}
	\renewcommand\thefigure{\Alph{section}\arabic{figure}}
	\renewcommand\thetable{\Alph{section}\arabic{table}}
}

\newcommand{\NP}{{\ensuremath{\mathsf{NP}}}\xspace}

\newcommand{\nats}{\mathbb{N}}

\newcommand{\Pol}{\mathrm{Pol}}
\newcommand{\PCSP}{\mathrm{PCSP}}
\newcommand{\CSP}{\mathrm{CSP}}

\newcommand{\B}{\mathbf{B}}
\newcommand{\A}{\mathbf{A}}

\newcommand{\LO}{\mathbf{LO}}

\newcommand{\Z}{\mathbf{Z}}
\newcommand{\X}{\mathbf{X}}

\title{Approximating 1-in-3 SAT by linearly ordered hypergraph $3$-colouring is \NP-hard
\thanks{This work is supported the UK EPSRC grant EP/X033201/1.}
}
\author{Andrei Krokhin\\
Durham University\\
\texttt{andrei.krokhin@durham.ac.uk} \and Danny Vagnozzi\\
Durham University\\
\texttt{danny.vagnozzi@durham.ac.uk}}

\date{}

\begin{document}

\maketitle

\begin{abstract}
    Given a satisfiable instance of 1-in-3 SAT, it is \NP-hard to find a satisfying assignment for it, but it may be possible to efficiently find a solution subject to a weaker (not necessarily Boolean) predicate than `1-in-3'.
    There is a folklore conjecture predicting which choices of weaker predicates lead to tractability and for which the task remains \NP-hard. One specific  predicate, corresponding to the problem of linearly ordered $3$-colouring of 3-uniform hypergraphs, has been mentioned in several recent papers as an obstacle to further progress in proving this conjecture. We prove that the problem for this predicate is \NP-hard, as predicted by the conjecture.

    We use the Promise CSP framework, where the complexity analysis is performed via the algebraic approach, by studying the structure of polymorphisms, which are multidimensional invariants of the problem at hand. The analysis of polymorphisms is in general a highly non-trivial task, and topological combinatorics was recently discovered to provide a useful tool for this. There are two distinct ways in which it was used: one is based on variations of the Borsuk-Ulam theorem, and the other aims to classify polymorphisms up to certain reconfigurations (homotopy). Our proof, whilst combinatorial in nature, shows that our problem is the first example where the features behind the two uses of topology appear together. Thus, it is likely to be useful in guiding further development of the topological method aimed at classifying Promise CSPs. An easy consequence of our result is the hardness of another specific Promise CSP, which was recently proved by Filakovsk\'y et al.~by employing a deep topological analysis of polymorphisms. 
     
\end{abstract}

\section{Introduction}

The Constraint Satisfaction Problem (CSP) framework is much studied in theoretical computer science (see e.g. \cite{kz17}).
Since the resolution of the Feder-Vardi CSP dichotomy conjecture \cite{bulatov2017,zhuk2017}, there has been much interest in the study of \emph{qualitative} approximations of \NP-hard CSPs. The notion of \emph{Promise} CSPs (PCSPs) provides the formal framework for studying these problems \cite{agh2017,bbko21,bg2016,ko2022}. In the PCSP paradigm, approximation is understood not with the aim of maximising the number of satisfied constraints (as, e.g., in \cite{has01,mm17}), but via relaxing, in a uniform way, the constraints that must be satisfied. The \emph{approximate graph colouring problem} is a prime example of a PCSP: for fixed constants $3\leq k\leq \ell$, given a graph promised to be $k$-colourable (but a $k$-colouring is not given), the goal is to {\em find} a $\ell$-colouring. While it is generally believed that this problem is \NP-hard for all constants $3\leq k\leq \ell$, there are only a few unconditional results on this matter. The state of the art is that 
the problem is \NP-hard for $k\ge 3$ and $\ell\le 2k-1$ (see \cite{bbko21}) and for 
$k\ge 4$ and $\ell\le \binom{k}{\lfloor k/2\rfloor}-1$ (see \cite{kowz2023}), while the best known efficient algorithm to colour a 3-colourable $n$-vertex graph requires $O(n^{0.19747})$ colours \cite{kty24}.
There are also many results about approximate hypergraph colouring (involving several different notions of colouring), see e.g.~\cite{abp20,agh2017,bbb21,bg16,drs05,fnotw2024,gs20b,nvwz2025,nz2023}.

There are three classical \NP-hard CSPs that play a big role in CSP theory \cite{kz17}: graph $k$-colouring with $k\ge 3$ and (the monotone versions of) Not-All-Equal (NAE) SAT and 1-in-3 SAT. The last two problems can also be seen as colouring problems on 3-uniform hypergraph using 2 colours: respectively, requiring  that no hyperedge is left monochromatic in the former case and that every hyperedge has exactly one vertex coloured as `1' in the latter case. Interestingly, the state-of-the-art for these three problems in the PCSP framework is very different.

Let us first describe in more detail the sort of problems we study in the present paper - formal definitions will be given in Section \ref{sec:defs}.
Given a satisfiable instance $I$ of 1-in-3 SAT, it is \NP-hard to find a satisfying assignment for it. 
We can try to relax the constraints 
- for example, by expanding the domain of values for the variables beyond $\{0,1\}$ -
in order to ease the task of finding a satisfying assignment for $I$ subject to the relaxed constraints.
More precisely, we fix a finite set $A$ and a ternary relation $R$ on $A$, such that $\{0,1\}\subseteq A$ and $R$ contains the relation/predicate of 1-in-3 SAT; that is, $\{(0,0,1),(0,1,0),(1,0,0)\}$. Given a satisfiable instance $I$ of 1-in-3 SAT, the goal is to find an assignment of elements from $A$ to the variables of $I$ so that every constraint in $I$ is satisfied in $R$. Put otherwise, we aim to find an `$R$-solution' to $I$, which can be seen as an `$R$-approximation to a 1-in-3 solution'.
The complexity of this task depends on $R$, and the question to find the exact dependence was first asked in \cite{bbko21}. As observed in \cite{bbb21}, one can assume without loss of generality that $R$ is symmetric; that is, invariant under all permutations of coordinates.
By analogy with approximate graph colouring, we call this family of problems {\em approximate 1-in-3 SAT}. Approximate NAE SAT is defined in a similar way, and it has actually been studied under the name of `approximate hypergraph colouring' \cite{drs05} before the introduction of the general PCSP framework. Approximate graph colouring can also be slightly expanded 
by using an arbitrary symmetric binary relation $R$ (containing the $k$-colouring relation) in place of the $\ell$-colouring relation.

Let us return to describing the status of the PCSPs approximating the three classical \NP-hard CSPs.
We call an $R$-approximation trivial if $R$ contains a loop, i.e.~a constant tuple. Obviously, trivial $R$-approximations are easy. For graph colouring, all non-trivial approximations are widely believed to be \NP-hard and this is known to hold assuming various UG-like conjectures \cite{bklm2022,dmr09,gs2020}. Without any such assumptions, this is known to hold only for a small number of cases discussed above. In contrast, all non-trivial approximations  for NAE SAT have been proven to be \NP-hard \cite{drs05,nvwz2025}. As for 1-in-3 SAT, there exist non-trivial approximations that are polynomial-time solvable. In fact, there is a folklore dichotomy conjecture (that we discuss in detail later) predicting which $R$-approximations are 
polynomial-time solvable and which are \NP-hard, and there are partial results confirming it in some special cases \cite{bbb21,bg2016,ckknz25}.

The main method in classifying the complexity of CSPs and PCSPs has been the algebraic approach, based on the notion of {\em polymorphisms}, which are multivariable functions capturing  invariant properties of the relations/predicates used to specify constraints (see surveys \cite{bkw17,ko2022}).
It is known that (certain abstract properties of) polymorphisms determine the complexity of CSPs and PCSPs \cite{bbko21,bop18,bg2016}. For CSPs, the presence of non-trivial (in a specific sense) polymorphisms implies tractability, while the lack thereof implies \NP-hardness \cite{bulatov2017,bjk05,zhuk2017}. In the PCSP realm, the situation is more complex as there are many known \NP-hard PCSPs with non-trivial (but limited in some sense) polymorphisms. As opposed to CSPs, it is not clear at all why general PCSPs should exhibit a $\mathsf{P}$ vs.~\NP-hard dichotomy. However, dichotomy is likely at least in some special cases. 

In general, analysing the polymorphisms of a CSP or a PCSP is a highly non-trivial task. A rule of thumb is that problems with `rich enough' polymorphisms are tractable (i.e.~polynomial-time solvable), whereas problems with `limited enough' polymorphisms are \NP-hard. For PCSPs (as opposed to CSPs), it is not clear yet what exactly `rich enough' and `limited enough' mean or whether there even is a gap between them (which would indicate a non-dichotomy). See \cite{bbko21} or survey \cite{ko2022} for further discussion of the general picture of PCSP research.

Unlike work towards the CSP dichotomy, where the right \NP-hardness condition was identified early in the process \cite{bjk05}, polymorphism-based methods for proving \NP-hardness of PCSPs are still very much in development \cite{bbko21,bk24,bk2022,bwz2021,do24}. Such methods typically involve identifying small special sets of variables in the polymorphisms. 
It was discovered recently that ideas from topological combinatorics \cite{mat03}
can be very useful in identifying these small sets. Such results have so far come in two distinct flavours.
One flavour uses reasoning with colourings of Kneser graphs or similar
(based on variants of the Borsuk-Ulam theorem) in order to
find a small set of variables in polymorphisms such that setting these variables to a particular value restricts the output of the function regardless of how the remaining variables are set (see \cite{abp20,bbko21,nvwz2025}). The other tries to show that each polymorphism is similar up certain changes (described in terms of homotopy transformations) to one with a small set of special variables, (see \cite{afotw25,fnotw2024,kowz2023,mo2025}). Curiously, the behaviours of polymorphisms underlying applications of the two topological flavours have never been observed in the same PCSP - until now. 



\subsection{Related work}
\label{sec:related}

Classifying the complexity of PCSPs is a very active area of research (see, e.g. survey \cite{ko2022}). We will give a brief overview that focus solely on the results that are closely related to approximating 1-in-3 SAT and to the topological method in PCSP.

Approximating 1-in-3 SAT by NAE was the first known tractable non-trivial approximation of a \NP-hard CSP by another \cite{bg2016}. More specifically, one algorithm works as follows: given a satisfiable instance of 1-in-3 SAT, one relaxes each constraint in it (say on variables $x,y,z$) to a linear equation $x+y+z=1$ over the integers. One then solves the obtained linear system over integers (which can be done in polynomial time \cite{kb79}), and rounds the obtained solution by mapping positive integers to 1 and non-positive ones to 0.
It is easy to see that this rounding sends any triple of integers such that $x+y+z=1$ to a triple of Boolean values that cannot be constant, thus providing a NAE solution for the original instance.
This polynomial-time algorithm is a special case of the basic AIP (affine integer programming) algorithm,
whose applicability for PCSPs was characterised in \cite{bbko21}. There is now a folklore conjecture (see Conjecture 5 in \cite{nz24}) predicting that  
every approximation of 1-in-3 SAT either can be solved by the above algorithm with appropriately chosen rounding, or it remains \NP-hard (see Section \ref{sec:pcsp} for a formal statement). This conjecture was confirmed  in \cite{ckknz25} for approximations by so-called rainbow-free relations (in which no tuple has three pairwise different elements).
The conjecture was also confirmed in \cite{bbb21} for relations on the set $\{0,1,2\}$, with a single exception that the authors were not able to classify.

This unclassified case from \cite{bbb21} is related to a variant of the  colouring problem for hypergraphs called  \emph{linearly ordered (LO) colouring}  \cite{bbb21,nz2023} (also known as the {\em unique maximum colouring} \cite{ckp13}). This variant assumes that the set of colours is linearly ordered, e.g.~equal to $\left\{0,1,\hdots,k-1\right\}$. A LO $k$-colouring of a hypergraph is an assignment of the colours to the vertices of a hypergraph so that, for each hyperedge $e$, if colour $c$ is the largest colour that appears in it, then $c$ is assigned to exactly one vertex in $e$.
For example, in an LO 3-colouring, for a hyperedge $(x,y,z)$, the colours $(0,1,2)$ or $(0,2,0)$ would be acceptable, but not $(0,1,1)$.
Observe that the LO 2-colouring problem for 3-uniform hypergraph is precisely 1-in-3-SAT.
Analogously to the approximate graph colouring problem, the \emph{approximate LO hypergraph colouring problem} is defined \cite{bbb21} as follows, for fixed constants $2\le k\le \ell$:
\begin{displayquote}
    Given a 3-uniform hypergraph that is promised to have a LO $k$-colouring, find a LO $\ell$-colouring for it.
\end{displayquote}

It was conjectured in \cite{bbb21} that this problem is \NP-hard for all
constants $2\le k\le \ell$. The case $k=2,\ell=3$ of this problem was exactly the only unclassified approximation of 1-in-3 SAT by a relation
on $\{0,1,2\}$. Up to now, the above conjecture was confirmed only in the following cases. It was observed \cite{fnotw2024} that, for $4\le k\le \ell$, there is a very easy reduction to the approximate LO colouring problems with constants $k,\ell$ from the approximate graph colouring problem with constants $k-1,\ell-1$. Hence, any known results 
about approximate graph colouring translate in this way to results about this conjecture. Apart from this, the only confirmed \NP-hardness case for approximate LO colouring is for $k=3,\ell=4$ \cite{fnotw2024}, with a proof based on topological 
combinatorics. LO hypergraph colouring can naturally be defined not only for 3-uniform hypergraphs, but for $r$-uniform hypergraph with $r\ge 3$. The approximate LO colouring problem is known to be \NP-hard for all $3\le k\le\ell$ and $r\ge 4$ and for $k=2$ and all $r\ge \ell+2$ (see \cite{nvwz2025,nz2023}).

The problem of finding a LO $3$-colouring for a given LO $2$-colourable 3-uniform hypergarphs (which is the same as the problem of approximating 1-in-3-SAT by LO 3-colourings)
was specifically mentioned  in \cite{bbb21,ckknz25,fnotw2024,nz2023,nvwz2025} as an open question and an obstacle for further progress in classifying the approximate 1-in-3 SAT problems.

\subsection{Our contributions}
We denote the problem of finding a LO $3$-colouring for a given LO $2$-colourable $3$-uniform hypergraph by $\PCSP\left(\LO_2,\LO_3\right)$ - the notation will be explained in Section \ref{sec:pcsp}.
In this paper, our main result settles its complexity.
\begin{theorem}\label{thm:main}
    $\PCSP\left(\LO_2,\LO_3\right)$ is \NP-hard.
\end{theorem}


We remark that the confirmed case $k=3,\ell=4$ of the LO colouring conjecture that we mentioned above easily follows from Theorem \ref{thm:main} (see Section \ref{sec:proof_main}), thus providing a more elementary proof of the main result in \cite{fnotw2024}.

The proof of Theorem \ref{thm:main} is a combination of a structural result (Theorem \ref{thm:structure}) followed by a direct application of a well-known sufficient condition for \NP-hardness of a PCSP \cite{bk2022,bwz2021}. Theorem \ref{thm:structure} states, roughly, that the polymorphisms of $\PCSP\left(\LO_2,\LO_3\right)$ exhibit two kinds of behaviour - precisely the two kinds that are characteristic for the two flavours of 
applications of topology in PCSPs that we described above. Even though our proof is combinatorial, it shows that our problem is the first example where the features behind the two uses of topology in PCSP appear together. Thus, it is likely to be very useful in guiding further development of the topological method in Promise CSP. It is important to note that whilst some of the applications of topology lead to alternative proofs of known results, e.g.~\cite{mo2025,nvwz2025}, or to a modest step forward in comparison with existing results, e.g.~\cite{afotw25,fnotw2024}, their methodological contribution is highly meaningful. The classification of PCSPs is still in its infancy and requires a far richer toolbox than for CSPs. The development of new methods or the identification of issues that said methods must be able to deal with are highly sought after by current research directions \cite{ko2022}.

\section{Preliminaries}
\label{sec:defs}

Throughout the paper we use the convention $\left[n\right]=\left\{1,\hdots,n\right\}.$  The complement of a set $X$, usually with respect to $\left[n\right]$, is denoted by $\overline{X}$.

\subsection{Promise CSPs}\label{sec:pcsp}

A \emph{relational structure} is a tuple $\mathbf{A}=\left(A;R_1,\hdots, R_m\right)$ where $A$ is a set referred to as the \emph{domain} of $\mathbf{A}$, and each $R_i$ is a relation on $A$ of arity $\mathrm{ar}_i\geq 1$; that is, $R_i\subseteq A^{\mathrm{ar}_i}$. In this paper we will only consider structures with a single symmetric binary relation (undirected graphs)
or with a single ternary symmetric relation (i.e. invariant under permutations of coordinates), which we refer to as \emph{$3$-uniform hypergraphs} by slightly abusing standard terminology.

Relational structures $\A=\left(A;R_1,\hdots, R_m\right)$ and $\B=\left(B; R'_1,\hdots,R'_{m'}\right)$ are said to be \emph{similar} if $m=m'$ and the relations $R_i$ and $R'_i$ have same arity for each $i\in\left[m\right].$ For similar relational structures $\A,\B$ as above, a mapping $f:A\rightarrow B$ is said to be a \emph{homomorphism} if  for all $i\in\left[m\right]$ we have $x\in R_i\implies f\left(x\right)\in R'_i$, where $f$ is applied component-wise to $x$. We write $f:\A\rightarrow \B$ to indicate that $f$ is a homomorphism, and use the notation $\A\rightarrow \B$ to indicate that a homomorphism from $\A$ to $\B$ exists.

\begin{definition}
    A \emph{PCSP template} is a pair $\left(\A,\B\right)$ where $\A,\B$ are similar relational structures such that $\A\rightarrow \B$. 
\end{definition}
 Every PCSP template $\left(\A,\B\right)$ defines a decision problem and a search problem as follows. The \emph{decision version} of $\PCSP\left(\A,\B\right)$ asks to decide, for a given input structure $\mathbf{X}$ similar to $\A$ and $\B$, whether $\mathbf{X}\rightarrow \A$ or whether $\mathbf{X}\not\rightarrow \B$, the promise being that exactly one of these cases occurs. The \emph{search version} of $\PCSP\left(\A,\B\right)$ asks, given an input structure $\mathbf{X}$ with a promise that $\mathbf{X}\rightarrow \A$, to find a homomorphism $f:\mathbf{X}\rightarrow \B$.  The decision version trivially reduces to the search version, but it is not known whether the two versions are always polynomial-time equivalent (see \cite{lar25} for recent results on this question). Observe that if $\A=\B$, then $\PCSP\left(\A,\B\right)$ is the same as the standard $\CSP\left(\A\right)$; that is, the problem of deciding whether $\mathbf{X}\rightarrow \A$. The PCSP framework allows a systematic study of problems that cannot be expressed as CSPs, and can be thought of as `qualitative' approximations of CSPs. Below are some  examples of PCSPs relevant for this paper - more can be found in \cite{abp20,agh2017,bbb21,bg16,drs05,fnotw2024,gs20b,ko2022,nvwz2025,nz2023}.

 \begin{exmp}[Approximate graph colouring] The problem asks to find, for fixed $\ell\geq k\geq 2$, an $\ell$-colouring of an input $k$-colourable graph. This problem is $\PCSP\left(\mathbf{K}_k,\mathbf{K}_\ell\right)$, where $\mathbf{K}_n$ is the $n$-clique; that is, $\mathbf{K}_n=\left(\left[n\right];\neq\right).$ 
 \end{exmp}
 \begin{exmp}[Approximate hypergraph colouring] This problem is a natural generalisation of the above, by replacing the inequality $\neq$ with the `not-all-equal' relation:
 $$
 \mathrm{NAE}_k=\left\{0,\hdots,k-1\right\}^3\backslash \left\{\left(a,a,a\right)\mid a\in\left\{0,\hdots,k-1\right\}\right\}.
 $$
 A $k$-colouring of a hypergraph, is an assignment of the colours $\left\{0,\hdots,k-1\right\}$ to the vertices such that no hyperedge is monochromatic. Hence, the problem of finding, for fixed $\ell\geq k\ge 2$, an $\ell$-colouring of a $k$-colourable $3$-uniform hypergraph is $\PCSP\left(\mathbf{H}_k,\mathbf{H}_\ell\right)$ where $\mathbf{H}_n=\left(\left\{0,\hdots,n-1\right\};\mathrm{NAE}_n\right)$.  
 Note that the classical \NP-hard problem NAE-SAT can be expressed as $\CSP\left(\mathbf{H}_2\right)$.
     
 \end{exmp}

 \begin{exmp}[Approximate linearly ordered (LO) hypergraph colouring]\label{ex:lo}
 This is stricter version of colouring than the one from the previous example, the requirement being 
that in each hyperedge, the maximal colour assigned to its vertices occurs exactly once. 
Formally we replace $\mathrm{NAE}_k$ with its subrelation
$$
\mathrm{LO}_k=\{(a,b,c)\in \left\{0,\hdots,k-1\right\}^3\mid (a,b,c) \mbox{ has a unique maximum}\}.
$$

For example, $\mathrm{LO}_2=\left\{\left(0,0,1\right),\left(0,1,0\right),\left(1,0,0\right)\right\}$, while $\mathrm{LO}_3$ contains exactly the triples $\left(0,1,2\right)$, $\left(1,1,2\right)$, $\left(0,0,1\right)$, $\left(0,0,2\right)$ and all permutations thereof, but not $\left(1,2,2\right)$ or $\left(0,1,1\right).$  Let $\LO_k=(\left\{0,\hdots,k-1\right\};\mathrm{LO}_k)$.
Then $\PCSP\left(\LO_k,\LO_\ell\right)$ is the Approximate LO hypergraph colouring problem that was discussed in detail in the previous section.
\end{exmp}

\begin{exmp}[Approximate 1-in-3 SAT]
The problem $\CSP\left(\LO_2\right)$ is precisely (monotone) 1-in-3 SAT, so
the class of problems of the form $\PCSP\left(\LO_2,\B\right)$, where the domain of $\B$ contains \{0,1\} and $\B$ has a unique ternary symmetric relation containing $\left\{\left(0,0,1\right),\left(0,1,0\right),\left(1,0,0\right)\right\}$ can naturally be called {\em Approximate 1-in-3 SAT}. The question of classifying the complexity of problems in this class was first asked in \cite{bbko21}.
 \end{exmp}

All known tractable cases of $\PCSP\left(\LO_2,\B\right)$ can be solved by essentially the same algorithm as the one for $\PCSP\left(\LO_2,\mathbf{H}_2\right)$ that we described in Section \ref{sec:related}. Consider the (infinite) structure $\Z=(\mathbb{Z}; \{(x,y,z)\mid x+y+z=1\})$. Trivially, we have $\LO_2\rightarrow \Z$ by inclusion. Assume that we additionally have $\Z\rightarrow \B$. Then $\PCSP\left(\LO_2,\B\right)$ can be solved in polynomial time by reducing it to solving a system of linear equations over $\mathbb{Z}$. Indeed, for an input $\X$, if $\X\rightarrow \LO_2$ then $\X\rightarrow\Z$. Finding a homomorphism from $\X$ to $\Z$ amounts to solving a system of linear equations over integers (which can done in polynomial time \cite{kb79}). By composing the found homomorphism from $\X$ to $\Z$ with the assumed homomorphism from $\Z$ to $\B$, we can find the required homomorphism from $\X$ to $\B$.
(There is a potential issue that homomorphisms from $\Z$ to $\B$ may not be computable, but this has not occurred so far). 
It has been conjectured independently by several authors (e.g. Conjecture 5 in \cite{nz24}) that 
$\PCSP\left(\LO_2,\B\right)$ is \NP-hard in all other cases.


\begin{conjecture}[The approximate 1-in-3 SAT conjecture]
For each $\B$, either $\Z\rightarrow \B$, in which case $\PCSP\left(\LO_2,\B\right)$ is solvable in polynomial time, or else $\PCSP\left(\LO_2,\B\right)$ is \NP-hard.
\end{conjecture}

As explained above, this conjecture has been confirmed in several special cases \cite{bbb21,ckknz25}. It is easy to verify that $\mathbf{Z}\not\rightarrow \LO_k$ for any $k\ge 2$, so that, indeed, $\PCSP\left(\LO_2,\LO_k\right)$ stands on the hardness side of the approximate 1-in-3 SAT conjecture. The simple reasoning, which can be found, e.g., in (the proof of) Proposition 8 in \cite{ckknz25} and which we now sketch, implies the following necessary condition for $\Z\rightarrow \B$. Let $D(\B)$ be the digraph whose vertices are the elements of $B$ and edges are all pairs $(x,y)$ such that $(x,x,y)$ is in the relation of $\B$. Assume that there is a homomorphism $\Z\rightarrow \B$. Consider all triples of the form $\left(a,a,1-2a\right)$ from $\Z$, and their images in $\B$ under the assumed homomorphism. These images show that the digraph $D(\B)$ must contain a homomorphic image of an infinite directed path, which always contains a directed cycle or a loop, since $B$ is finite. It is easy to see that $D(\LO_k)$ is a transitive tournament, so it contains neither a directed cycle nor a loop, and it follows that $\mathbf{Z}\not\rightarrow \LO_k$. We remark that the structures $\LO_k$ are the richest structures $\B$ with the property that $D(\B)$ contains neither a directed cycle nor a loop - it is not hard to see that any ($k$-element) structure with this property is isomorphic to a substructure of $\LO_k$.



\subsection{Polymorphisms}

Polymorphisms are ``higher order'' homomorphisms that determine the complexity of CSPs and PCSPs. The body of work concerned with the use of polymorphisms for deriving reductions between PCSPs is generally referred to as the algebraic theory of PCSPs, the core of which can be found in \cite{bbko21,bg2016}. In this section we review some basic notions. 


\begin{definition}
Assume that $(\A,\B)$ is a PCSP template. 
A $n$-ary \emph{polymorphism} of a template $\left(\A,\B\right)$ is a function $f:A^n\rightarrow B$ such that,  for every $i$ and  for every $\mathrm{ar}_i\times n$ matrix $M$ with entries from $A$, if all columns of $M$ belong to the relation $R_i$ in $\A$, then the column obtained by applying $f$ to the rows of $M$ belongs to the corresponding relation $R'_i$ in $\B$. If $\A=\B$, then we say that  $f$ is a polymorphism of $\A$.
\end{definition}

If $(\A,\B)$ is any PCSP template and $\varphi: \A\rightarrow \B$, then it is easy to see that, for all $n\ge 1$ and $i\in\left[n\right]$,
all functions of the form $f\left(x_1,\hdots,x_n\right)=\varphi(x_i)$  are polymorphisms of
$(\A,\B)$. In the special case when $A\subseteq B$ and $\varphi$ is the inclusion map, such polymorphisms are called {\em projections} (also known as dictators). 

\begin{exmp}[Polymorphisms of $\LO_2$]\label{ex:pollo2}
A function $f:\{0,1\}^n\rightarrow \{0,1\}$ is a polymorphism of $\LO_2$ if, for every $3\times n$ matrix $M$ with 0/1 entries such that every column of $M$ contains exactly one 1, 
the column obtained by applying $f$ to the rows of $M$ also contains exactly one 1.
It is well-known and not hard to verify directly, that projections are the only polymorphisms of $\LO_2$.
\end{exmp}
\begin{exmp}[Polymorphisms of $\left(\LO_2,\mathbf{H}_2\right)$] 
A function $f:\{0,1\}^n\rightarrow \{0,1\}$ is a polymorphism of $\left(\LO_2,\mathbf{H}_2\right)$ if, for every $3\times n$ matrix $M$ with 0/1 entries, if every column of $M$ contains exactly one 1, then 
the column obtained by applying $f$ to the rows of $M$ contains both a 0 and a 1. This template has many polymorphisms different from projections. 

One example is the $(3n+1)$-ary Boolean function $f'$  which outputs 1 if and only if at least $n+1$ of its arguments are equal to 1. To see that this is a polymorphism, consider a $3\times (3n+1)$ matrix $M$ such that each column of $M$ contains exactly one 1. Since the total number of 1s in $M$ is $3n+1$, it is impossible that each row of $M$ contains at least $n+1$ 1s or that each row of $M$ contains $n$ or fewer 1s.
Therefore $f'$ applied to the rows of $M$ produces a column that contains both a 0 and a 1.
\end{exmp}
The polymorphisms of a structure $\A$ are closed under composition; hence the success of applications of universal algebra theory in the classification of CSPs. This is not the case for PCSPs: indeed, for a general template, composition of polymorphisms is not well defined. Instead, we work with a `composition free' structure where the functions are closed under permuting the variables, identifying them or adding dummy ones. Formally, we refer to these operations as taking a \emph{minor}.

\begin{definition}
    Let $f:A^n\rightarrow B$ be a function. For a map  $\pi:\left[n\right]\rightarrow \left[m\right]$, we say that a function $g:A^m\rightarrow B$ is a \emph{minor} (or a $\pi$-minor) of $f$, denoted $g=f^\pi$, if
    $$
    g\left(a_1,\hdots,a_m\right)=f\left(a_{\pi\left(1\right)},\hdots,a_{\pi\left(n\right)}\right)
    $$
    for all $a_1,\hdots,a_m\in A$.

    A set of finitary functions from $A$ to $B$ that is closed under taking minors is called a \em{(function) minion}.
\end{definition}

It is easy to see that, for any PCSP template $(\A,\B)$, the set of all polymorphisms of
$(\A,\B)$ is a minion. It is denoted by $\Pol\left(\A,\B\right)$.

There are several polymorphism-based sufficient conditions for $\NP$-hardness of PCSPs (see, e.g., \cite{bk24,bbko21,bk2022,bwz2021,ko2022}). We will use the following one in the proof of Theorem \ref{thm:main}. A \emph{chain of minors} in $\Pol\left(\A,\B\right)$ is a sequence $\left(f_1,\pi_{1,2},f_2,\hdots,f_{l-1},\pi_{l-1,l},f_l\right)$ where $f_i\in \Pol\left(\A,\B\right)$ for each $i\in\left[l\right]$ and $f_i^{\pi_{i,i+1}}=f_{i+1}$.

\begin{theorem}[\cite{bk2022,bwz2021}]\label{thm:babypcp}
    Fix $k,l\in\nats$ and let $\left(\A,\B\right)$ be a PCSP template. Suppose that to each polymorphism $f\in\Pol\left(\A,\B\right)$ we assign a subset $I\left(f\right)$ of at most $k$ of its variables. Suppose further that, for any chain of minors $\left(f_1,\pi_{1,2},f_2,\hdots,f_{l-1},\pi_{l-1,l},f_l\right)$ in $\Pol\left(\A,\B\right)$, there exist $1\leq i<j\leq l$ such that $\pi^{-1}_{i,j}\left(I\left(f_j\right)\right)\cap I\left(f_i\right)\neq \emptyset$ where $\pi_{i,j}=\pi_{j-1,j}\circ \hdots \circ \pi_{i,i+1}$. Then $\PCSP\left(\A,\B\right)$ is \NP-hard.
\end{theorem}

The above theorem was first proved in \cite{bwz2021} by using the PCP theorem and the NP-hardness of (layered) Gap Label Cover problem. It was later shown in \cite{bk2022} that the \NP-hardness in fact follows by a direct and simple ``bounded width'' reduction from 1-in-3 SAT (or from any $\NP$-hard CSP).


\section{The structure of $\Pol\left(\LO_2,\LO_3\right)$}
The core of this section is dedicated to the proof of the structural result (Theorem \ref{thm:structure}) describing the general form of the polymorphisms of $\left(\LO_2,\LO_3\right)$ and to showing how this result implies Theorem \ref{thm:main}.

\subsection{$0$-sets, $1$-sets, $2$-sets...}

We use the notion of an $i$-\emph{set} as in~\cite{bbb21} to describe the action of a polymorphism of $\left(\LO_2,\LO_3\right)$. Fix some $n$ and consider a function 
$f:\left\{0,1\right\}^n\rightarrow \left\{0,1,2\right\}$. We can identify tuples from $\{0,1\}^n$ with subsets from $[n]$, where a subset $X$ corresponds to the tuple having 1 exactly in positions in $X$. Then notation $f(X)$ has the obvious meaning. We say that $X\subseteq\left[n\right]$ is an $i$-set of $f$ for some $i\in\left\{0,1,2\right\}$ if $f(X)=i$. It is easy to see that if $g=f^\pi$ and $X$ is an $i$-set of $g$ then 
$\pi^{-1}(X)$ is an $i$-set of $f$. If $f(X)\in \left\{0,1\right\}$ then we call $X$ a \emph{boolean set} of $f$, and we say that
a 2-set of $f$ {\em small} if it contains three elements or less. 


It is straightforward to check that a function $f:\left\{0,1\right\}^n\rightarrow \left\{0,1,2\right\}$ is a polymorphism of $\left(\LO_2,\LO_3\right)$ if, and only if, for all partitions $\left(X,Y,Z\right)$ of $\left[n\right]$ into three subsets (some of which may be empty), we have that $\left(f\left(X\right),f\left(Y\right),f\left(Z\right)\right)\in \mathrm{LO}_3$. Observe that the empty set can never be a 2-set.
We say that a partition $\left(X,Y,Z\right)$ of $\left[n\right]$ is a \emph{boolean partition} (for $f$) if $\left\{f\left(X\right),f\left(Y\right),f\left(Z\right)\right\}=\left\{0,1\right\}$. In the language of $i$-sets, the $n$-ary projections of $\left(\LO_2,\LO_3\right)$ are the mappings for which there exists some $t\in\left[n\right]$, such that $X\subseteq\left[n\right]$ is a $1$-set if, and only if, $t\in X$ and it is a $0$-set otherwise.

When we use the above terminology ($i$-set, boolean partition etc.)~and $f$ is clear from the context, we will omit a reference to it. This should cause no confusion.

\subsection{Reconfiguration graph of polymorphisms}

Let us informally describe our result about the structure of $\Pol\left(\LO_2,\LO_3\right)$ by means of the so-called \emph{reconfiguration graph}. For $n\ge 1$, consider a graph $G_n$ whose vertices are the $n$-ary polymorphisms of $\left(\LO_2,\LO_3\right)$ and there is an edge between two polymorphisms if they differ (as functions) in a single position. Our structural result (Theorem \ref{thm:structure}) claims, intuitively, that 
if all polymorphisms with small $2$-sets are removed from $G_n$, then the resulting graph is disconnected, and each of its connected components contains a unique $n$-ary projection. For our hardness result, we will need a more precise and technical version of this statement.

The split between polymorphisms with small $2$-set and the rest can be justified from the fact that the absence of a small $2$-set is a sufficient condition for a polymorphism to have a very rigid structure that mimics a projection up to some reconfiguration (indeed, the unique projection in its connected component of $G_n$ once the polymorphisms with small $2$-sets are removed). This is also guaranteed by the following consequence of Lov\'asz's result \cite{lovasz} about the chromatic number of the Kneser graphs (a founding pillar of the area of topological combinatorics).

\begin{lemma}[Lemma 13 in~\cite{bbb21}]\label{lem:kneser}
    Any polymorphism of $\left(\LO_2,\LO_3\right)$ has a 1-set or a 2-set with at most $3$ elements.
\end{lemma}

\subsection{Recolouring and saturation}
The $n$-ary polymorphisms of $\left(\LO_2,\LO_3\right)$ can be understood as colourings of the subsets of $\left[n\right]$ using the colour set $\left\{0,1,2\right\}$ and subject to the constraints enforced by $\LO_3$. We say that a set $X\subseteq\left[n\right]$ is \emph{recolourable} for a $n$-ary polymorphism $f$ if one can change the image of $X$ to a different value to obtain a mapping that is still a polymorphism of $\left(\LO_2,\LO_3\right)$. Specifically, we say that $X$ is \emph{recolourable to $i$} or $i$-\emph{recolourable} for some $i\neq f\left(X\right)$ if the mapping $f'$ defined as
$$
f'\left(Y\right)=\begin{cases}
    f\left(Y\right)\;\;\textrm{if}\;Y\neq X\\
    i\;\;\;\;\;\;\;\;\;\textrm{otherwise}
\end{cases}
$$
is a polymorphism of $\left(\LO_2,\LO_3\right)$. We say that a boolean set is \emph{static} (for $f$) if it is not recolourable to its opposite boolean value. That is, if it is an $i$-set for $i\in\left\{0,1\right\}$ and it is not $\left(1-i\right)$-recolourable. The following lemma is straightforward.

\begin{lemma}
\label{lem:static}
    A boolean set $X$ of a $n$-ary polymorphism $f$ of $\left(\LO_2,\LO_3\right)$ is static if, and only if, there exist sets $Y,Z$ such that $(X,Y,Z)$ is a boolean partition of $[n]$.  
\end{lemma}

Recolouring is a central notion to our structural result: what we show is in fact, that polymorphisms of $\left(\LO_2,\LO_3\right)$ without small $2$-sets are essentially projections modulo some recolouring. 

\begin{definition}
    Let $f$ be a $n$-ary polymorphism of $\left(\LO_2,\LO_3\right)$. We say that $f$ is a \emph{recoloured projection} if there is some $t\in\left[n\right]$ such that for any static boolean set $S$, it holds that $f\left(S\right)=\left[t\in S\right]$ where $\left[t\in S\right]$ is the Iverson bracket (which evaluates to $1$ if $t\in S$ and to $0$ otherwise). We refer to such $t\in\left[n\right]$ as the \emph{dictating variable} of $f$.
\end{definition}

It is easy to characterise all $2$-recolourable boolean sets. The following is straightforward.
\begin{lemma}\label{lem:2set}
    For any polymorphism $f$ of $\left(\LO_2,\LO_3\right)$, there are no disjoint $2$-sets. In particular, the empty set is never a $2$-set. Furthermore, a boolean set $X$ is $2$-recolourable if, and only if, $X$ has non-empty intersection with every $2$-set of $f$.
\end{lemma}

The above lemma states that the 2-sets of any polymorphism form an intersecting family
(i.e. a family of pairwise intersecting subsets of a set), which is an object much studied in extremal combinatorics, see,  e.g. \cite{ellis21}.
We say that a $n$-ary polymorphism $f$ is \emph{upwards closed} if for every $2$-set $X$, we have that $X\subseteq Y$ implies $f\left(Y\right)=2$. If $f$ is upwards closed and is such that the complement of every boolean set is a $2$-set (that is, for every $X\subseteq\left[n\right]$ either $X$ or $\overline{X}$ is a $2$-set), we say that $f$ is \emph{saturated}. Note that in a saturated polymorphism no boolean set is $2$-recolourable, so the $2$-sets of a saturated polymorphism of arity $n$ form a maximal intersecting family of subsets, as studied in \cite{meyerovitz95}. 

\begin{definition}
    Let $f,g$ be $n$-ary polymorphisms of $\left(\LO_2,\LO_3\right)$ and assume that $g$ is saturated. We say that $g$ is a \emph{saturation} of $f$ if there is a sequence $f_1=f,f_2,\hdots,f_{l-1},f_l=g$ of $n$-ary polymorphisms where for each $i\in\left[l\right]$, $f_i$ is obtained from $f_{i-1}$ by recolouring a boolean set to $2$. If, in addition, $g$ has no small $2$-set, we say that $g$ is a \emph{pure saturation} of $f$. We call the sequence $f_1,\hdots,f_l$ a \emph{saturation path} for $f$. 
\end{definition}

When using the fact that a polymorphism $f$ is saturated in a proof, we may say `by saturation' as a shorthand for `because $f$ is saturated' when $f$ is clear from the context. The following lemma implies that every polymorphism can be recoloured to one that is upwards closed.

\begin{lemma}\label{lem:upwards_closure}
    Let $f$ be a polymorphism of $\left(\LO_2,\LO_3\right)$. If $X$ is a $2$-set and $Y$ is a boolean set such that $X\subseteq Y$, then $Y$ is $2$-recolourable.
\end{lemma}
\begin{proof}
    By Lemma \ref{lem:2set}, $X$ has non-empty intersection with all $2$-sets of $f$. Since $X\subseteq Y$, the same holds for $Y$, so $Y$ is $2$-recolourable by Lemma \ref{lem:2set}.
\end{proof}

Using a similar argument, we can show that every polymorphism has a saturation, though this may not be unique.

\begin{lemma}\label{lem:complementarity}
    Let $f$ be an upwards closed $n$-ary polymorphism of $\left(\LO_2,\LO_3\right)$. If $X,Y$ are disjoint non-empty boolean sets for which $X\cup Y=\left[n\right]$, then both $X$ and $Y$ are $2$-recolourable.
\end{lemma}
\begin{proof}
    Since $f$ is upwards closed, there are no subsets of $X$ that are $2$-sets. Thus, all $2$-sets intersect $Y$, so by Lemma \ref{lem:2set}, $Y$ is $2$-recolourable. By symmetry, the same argument applies to $X$.
\end{proof}

\begin{corollary}\label{cor:lift}
Every polymorphism $f$ of $\left(\LO_2,\LO_3\right)$ has a saturation. Moreover, if $f$ has arity at least 7 and no small $2$-sets, then $f$ has a pure saturation.
    
\end{corollary}

\begin{proof}
    Starting from $f$, one can apply Lemmas \ref{lem:upwards_closure} and \ref{lem:complementarity} exhaustively to obtain a saturation.
    If $f$ has arity $n\ge 7$ and no small 2-set, then, when applying Lemma \ref{lem:complementarity}, one can always choose to 2-recolour the set containing at least $\lceil n/2\rceil$ elements, thus never introducing small 2-sets.
\end{proof}

Observe that the transformation from $f$ to its saturation can be seen as taking the intersecting family of all 2-sets of $f$, completing it to a maximal intersecting family, and then 2-recolouring all boolean sets in this maximal family.

Being able to recolour a polymorphism to a saturated one is a handy trick: this gives more structure and makes analysing recolourings easier without affecting the core of the polymorphism itself. 
For saturated polymorphisms, there is a more useful characterisation of static boolean sets. Hereafter, for a saturated polymorphism $f$, we denote the union of its minimal $2$-sets by $T_f$ \footnote{We understand minimality of $2$-sets by inclusion: a $2$-set is minimal if none of its non-trivial subsets is a $2$-set.}.

\begin{lemma}\label{lem:boolean_recolouring}
    For an $n$-ary saturated polymorphism $f$ of $\left(\LO_2,\LO_3\right)$, a boolean set $X$ is static if, and only if, $X\cap T_f$ is non-empty.
\end{lemma}
\begin{proof}
    Suppose $X\cap T_f$ is empty, and consider a partition $\left(X,U,V\right)$ of $\left[n\right]$. If $U$ is boolean then $\overline{U}$ is a $2$-set by saturation and thus, there must be some minimal $2$-set $Q$ for which $Q\subseteq \overline{U}$. In particular, $\overline{U}\cap T_f$ is non-empty. Since $\overline{U}$ is the disjoint union of $X$ and $V$ and $X$ is disjoint from $T_f$ we must have that $\overline{U}\cap T_f\subseteq V$. This implies that $Q\subseteq V$ so $V$ is a $2$-set, so $X$ is not part of any boolean partition and hence, it is not static by Lemma \ref{lem:static}.

    For the converse, let $Q$ be a minimal $2$-set such that $X\cap Q$ is non-empty. Consider the partition $\left(X,Q\backslash X, \overline{Q\cup X}\right)$. It is boolean since $X$ is boolean by assumption, $Q\backslash X$ is a proper subset of a minimal $2$-set, and $\overline{Q\cup X}$ is disjoint from a $2$-set, namely $Q$. Thus, $X$ is a part of a boolean partition, so it is static by Lemma \ref{lem:static}.
\end{proof}

As mentioned above, polymorphisms may not have a unique saturation - projections being the most immediate example. For polymorphisms with unique pure saturation, we have the following characterisation of the boolean sets that are recoloured throughout a saturation path.

\begin{lemma}\label{lem:unique_saturation2}
    Let $f$ be a polymorphism of $\left(\LO_2,\LO_3\right)$ of arity $n\geq 7$ and no small $2$-sets. Let $X$ be a boolean set of $f$ and let $g$ be a pure saturation of $f$. If $\left|X\right|\geq n-3$ or there is some $Z\subset X$ such that $f\left(Z\right)=2$, then $g\left(X\right)\neq f\left(X\right)$. If, in addition, we assume that $g$ is the unique pure saturation of $f$, then the converse holds.
\end{lemma}

\begin{proof}
     Let $X$ be a boolean set of $f$. If $\left|X\right|\geq n-3$, then $g\left(X\right)=2$, since $g$ cannot have small $2$-sets and either $X$ or $\overline{X}$ must be a $2$-set of $g$. If $X$ has a subset that is a $2$-set of $f$, then $g\left(X\right)=2$, since saturations are upwards closed. In both cases, $g\left(X\right)\neq f\left(X\right)$ as required.

    For the converse, let $X$ be a boolean set of $f$ and suppose $g\left(X\right)\neq f\left(X\right)$, where $g$ is the unique pure saturation of $f$. Note that this implies $g\left(X\right)=2$. Suppose that $\left|X\right|\leq n-4$ and no subset of $X$ is a $2$-set of $f$. By Lemma \ref{lem:2set}, $\overline{X}$ must be $2$-recolourable in $f$, so let $f'$ be the polymorphism obtained by $2$-recolouring $\overline{X}$. Since $\left|\overline{X}\right|\geq 4$, $f'$ has no small $2$-sets; in particular, there is a saturation $g'$ of $f'$ with no small $2$-sets (by Corollary \ref{cor:lift}) such that $g'\left(X\right)\neq 2$, which means that $g'\left(X\right)=f\left(X\right)\neq g\left(X\right)$. But $g'$ is by construction a pure saturation of $f$, which contradicts the assumption that $g$ is unique.   
\end{proof}

We now state the structural theorem at the core of our main result (Theorem \ref{thm:main}). We provide the proof in Section \ref{sec:proof_of_structure}. 

\begin{theorem}\label{thm:structure}
    For a $n$-ary polymorphism $f$ of $\left(\LO_2,\LO_3\right)$ that has no small 2-sets, there exists a unique $t\in\left[n\right]$ (depending on $f$) such that the following holds: $n\le 6$ and $f$ is a recoloured projection with dictating variable $t$, or $n\ge 7$ and every pure saturation of $f$ is a recoloured projection with dictating variable $t$.
\end{theorem}



\subsection{Proof of Theorem \ref{thm:main} and the hardness of $\PCSP\left(\LO_3,\LO_4\right)$ revisited}\label{sec:proof_main}

In the section, we show how Theorem \ref{thm:structure}, with the help of Theorem \ref{thm:babypcp} and some auxiliary lemmas, implies Theorem \ref{thm:main}. In this proof, we will need to distinguish cases depending on whether or not polymorphisms under consideration have a unique pure saturation, and the three technical lemmas below provide useful properties of these cases.






\begin{lemma}\label{lem:non-unique}
Assume that $f$ is a $n$-ary polymorphism of $\left(\LO_2,\LO_3\right)$ with $n\geq 7$ and no small 2-sets. Suppose that $f$ does not have a unique pure saturation and let $t$ be the dictating variable of any pure saturation of $f$, as in Theorem \ref{thm:structure}. Then $f$ cannot have disjoint non-empty boolean sets $S,T\subseteq\left[n\right]$ such that $f(S)\ne [t\in S]$ and $f(T)\ne [t\in T]$.
\end{lemma}

\begin{proof}
    Let $g$ be a pure saturation of $f$. Since the choice of $g$ is not unique, there is a minimal 2-set $A$ of $g$ such that neither $A$ nor $\overline A$ is a 2-set of $f$. Clearly, we have $|A|\ge 4$. It is easy to see that one can choose a saturation path from $f$ to $g$ so that 2-recolouring of $A$ is the last step on this path.  
    Suppose, for contradiction, that there exist disjoint non-empty $S,T$ that are boolean sets of $f$ for which $f\left(S\right)\neq \left[t\in S\right]$ and $f\left(t\right)\neq\left[t\in T\right]$. We cannot have $g(S)=g(T)=2$, since $S$ and $T$ are disjoint, so assume without loss of generality that $g(S)\ne 2$, and hence $g(S)=f(S)$. To get a contradiction, it suffices to show that $S$ is static for $g$, because then $g(S)=[t\in S]$ by Theorem \ref{thm:structure} and so $g(S)\ne f(S)$. By Lemma \ref{lem:boolean_recolouring}, it suffices to show that $S\cap T_g\ne \emptyset$.
    
    Consider first the case $|\overline A|\le 3$. Choose any element $a\in A$.
    Since $A$ is a minimal 2-set of $g$ and $g$ is saturated, we have that $A\backslash \{a\}$ is a boolean set of $g$, and its complement $\overline{A}\cup \{a\}$ is a 2-set of $g$, which must be minimal because it contains at most four (and hence exactly four) elements. Since both $A$ and $\overline{A}\cup \{a\}$ are minimal 2-sets of $g$, we have that $T_g=[n]$ and $S\cap T_g\ne \emptyset$.

    Assume now that $|\overline A|\ge 4$ and suppose, for contradiction, that $S\cap T_g=\emptyset$. Then $S\subseteq \overline{A}$. Suppose we have that $S\subset \overline{A}$ and consider the saturation $g'$ of $f$ such that $g'(X)=g(X)$ if $X\notin\left\{ A,\overline{A}\right\}$, $g'(A)=f(A)$ and $g'(\overline{A})=2$. Put otherwise, we follow the saturation path from $f$ to $g$, but in the last step we make a different choice: we 2-recolour $\overline{A}$ rather than $A$. Note that, since $|\overline A|\ge 4$, $g'$ is a pure saturation of $f$ and $\overline{A}$ is its minimal 2-set. Since $S\subset \overline{A}$, $S$ is a boolean set of $g'$, which is static by Lemma \ref{lem:boolean_recolouring}. But $g'(S)=f(S)\ne [t\in S]$, and this contradicts $t$ being a dictating variable of $g'$. So we must have that $S=\overline{A}$. Since $g$ is saturated and has no small 2-sets, it has a minimal 2-set $B$ distinct from $A$. Then $S\cap B\neq \emptyset$, and thus $S$ is a static boolean set of $g$.  
 \end{proof}

 \begin{lemma}\label{lem:hitting_set_projections}
     Let $f$ be a polymorphism of $\left(\LO_2,\LO_3\right)$ with arity $n\geq 7$ and no small $2$-set. Suppose further that $f$ has a unique pure saturation $g$. Let $t$ be the dictating variable of $g$, and let $X\subseteq\left[n\right]$ be any set such that $t\in X$ and $f\left(X\right)=1$. Then $X\cap Y\neq \emptyset$ for any set $Y\subseteq\left[n\right]$ such that $f\left(Y\right)=1$ and $Y\cap T_g\neq \emptyset$.
 \end{lemma}

\begin{proof}
    Let $Y$ be a $1$-set of $f$ and suppose $Y\cap T_g$ is non-empty. We distinguish the cases $g\left(Y\right)=1$ and $g\left(Y\right)=2$. If $g\left(Y\right)=1$, then $Y$ is static in $g$ by Lemma \ref{lem:boolean_recolouring}, so $t\in Y$ by Theorem \ref{thm:structure} and thus, $X\cap Y$ is non-empty as required.
    Consider now the case $g\left(Y\right)=2$. Because $g$ is the unique pure saturation of $f$ and $g\left(Y\right)\neq f\left(Y\right)$, it follows from Lemma \ref{lem:unique_saturation2} that either $\left|Y\right|\geq n-3$ or $Y$ has a subset $Z\subset Y$ such that $f\left(Z\right)=2$. Suppose $Y$ has a subset $Z$ for which $f\left(Z\right)=2$. If $X$ is disjoint from $Y$ then there must be a partition $\left(W,X,Y\right)$ of $\left[n\right]$ with $f\left(W\right)=2$, and therefore $W,Z$ are disjoint $2$-sets - a contradiction to Lemma \ref{lem:2set}. If instead, $\left|Y\right|\geq n-3$ and $X$ is disjoint from $Y$ then there is a partition $\left(W,X,Y\right)$ of $\left[n\right]$ where $f\left(X\right)=f\left(Y\right)=1$ so $f\left(W\right)=2$, but $\left|W\right|\leq 3$, which is impossible since $f$ does not have any small $2$-sets. Thus, $X\cap Y$ is non-empty as required.
\end{proof}

\begin{lemma}\label{lem:saturation_commutes_sometimes}
    Let $f$ be a polymorphism of $\left(\LO_2,\LO_3\right)$ with arity $n\geq 7$ and no small $2$-set. Suppose $g$ is the unique pure saturation of $f$. Let $f'=f^\pi$ be a $m$-ary polymorphism for some minor map $\pi:\left[n\right]\rightarrow \left[m\right]$ with $m\geq 7$ and suppose $f'$ has no small $2$-set. If $g'$ is the unique pure saturation of $f'$, and if $g'$ has no minimal $2$-set with $m-3$ or more elements, then $\pi\left(t\right)=t'$, where $t,t'$ are the respective dictating variables of $g,g'.$
\end{lemma}

Note that for a $m$-ary saturated polymorphism with no small $2$-set, all sets with $m-3$ or more elements are $2$-sets, so it suffices to show the above statement assuming $g'$ has no minimal $2$-sets with exactly $m-3$ elements. 
\begin{proof}
    First we show that $g'$ and $g^\pi$ coincide on all boolean sets. Let $X$ be a boolean set of $g'$, so we have $g'\left(X\right)=f'\left(X\right)$. Since $g'$ is the unique pure saturation of $f'$ and $g'\left(X\right)=f'\left(X\right)$, it follows from Lemma \ref{lem:unique_saturation2} that either $\left|\overline{X}\right|\geq m-3$ (or equivalently $\left|X\right|\leq 3$) or there is some $Z\subseteq\overline{X}$ such that $f'\left(Z\right)=2$. Consider first the case $\left|X\right|\leq 3$. Since $g'$ has no minimal $2$-set of size $m-3$ or larger, and $g'\left(\overline{X}\right)=2$, it must hold that there is some $Z\subset \overline{X}$ such that $Z$ is a $2$-set of $g'$ and $\left|Z\right|\leq m-4$. But $Z$ is therefore a $2$-set of $f'$, and thus, $\pi^{-1}\left(Z\right)$ is a $2$-set of $f$. Since $\pi^{-1}\left(X\right)$ is disjoint from $\pi^{-1}\left(Z\right)$, it follows from Lemma \ref{lem:2set} that $\pi^{-1}\left(X\right)$ is a boolean set of both $f$ and $g$ and, in particular, $f\left(\pi^{-1}\left(X\right)\right)=g\left(\pi^{-1}\left(X\right)\right)$. But $g'\left(X\right)=f'\left(X\right)=f^\pi\left(X\right)=f\left(\pi^{-1}\left(X\right)\right)$ and $g^\pi\left(X\right)=g\left(\pi^{-1}\left(X\right)\right)$, so $X$ is a boolean set of $g^\pi$ with $g^\pi\left(X\right)=g'\left(X\right)$. For the case where $Z\subseteq \overline{X}$, we proceed with an argument verbatim as above and conclude that $X$ is a boolean set of $g^\pi$ and $g^\pi\left(X\right)=g'\left(X\right)$. In particular, this means that for all boolean sets $X$ of $g'$, we have $g'\left(X\right)=g\left(\pi^{-1}\left(X\right)\right)$ as claimed. 

    Let $\left(W,X,Y\right)$ be a partition of $\left[m\right]$ that is boolean for $g'$. From the above argument it follows that $\left(W,X,Y\right)$ is also a boolean partition for $g^\pi$ and hence, $\left(\pi^{-1}\left(W\right),\pi^{-1}\left(X\right),\pi^{-1}\left(Z\right)\right)$ is a partition of $\left[n\right]$ that is boolean for $g$. In particular, for any static $1$-set $X$ of $g'$, it holds that $\pi^{-1}\left(X\right)$ a static $1$-set of $g$. By Theorem \ref{thm:structure}, we may choose static $1$-sets of $g'$, say $A$ and $B$, such that $A\cap B=\left\{t'\right\}$. Then since $\pi^{-1}\left(A\right),\pi^{-1}\left(B\right)$ are static $1$-sets of $g$ it follows from Theorem \ref{thm:structure} that $t\in \pi^{-1}\left(A\right)\cap \pi^{-1}\left(B\right)$. Thus, $t\in\pi^{-1}\left(t'\right)$ or, equivalently, $\pi\left(t\right)=t'$ as required.
\end{proof}

We now apply the properties proven in Lemmas \ref{lem:non-unique}, \ref{lem:hitting_set_projections} and \ref{lem:saturation_commutes_sometimes} to show Theorem \ref{thm:main}.

\begin{proof}[Proof of Theorem \ref{thm:main}]

We apply Theorem \ref{thm:babypcp} to $\Pol\left(\LO_2,\LO_3\right)$ with $k=3,l=4$.

For each polymorphism $f$ of $\left(\LO_2,\LO_3\right)$, we choose $I(f)$ to be an arbitrary small 2-set of $f$ if there is one. Otherwise, if $f$ has arity $n\leq 6$ then $f$ is a recoloured projection and we take $I\left(f\right)=\left\{t\right\}$ where $t$ is the dictating variable of $f$. If $f$ has arity $n\geq 7$ and no small $2$-set, consider a pure saturation $g$ of $f$ and let $t$ be its dictating variable. In this case, we take $I\left(f\right)=\left\{t\right\}$ if $f\left(\left\{t\right\}\right)=1$; otherwise, we take $I\left(f\right)$ to be any $2$-element $1$-set containing $t$ (the existence of such sets is an easy consequence of Theorem \ref{thm:structure}). It may seem that it would be more natural to set $l=3$ and to let $I(f)=\{t\}$ for all polymorphisms $f$ without small 2-sets. We will show in Example \ref{exmp:1} below that this would not work.

First, observe that, for any polymorphisms $f$ and $h$ such that $h=f^\pi$, the set $\pi^{-1}\left(I\left(h\right)\right)$ is never empty.
Indeed if $I\left(h\right)$ is a $2$-set, then $\pi^{-1}\left(I\left(h\right)\right)$ is a $2$-set of $f$, and so is non-empty by Lemma \ref{lem:2set}. Assume that the arity of $h$ is at least 2 and $I\left(h\right)$ contains the dictating variable $t$ of $h$ or that of a pure saturation of $h$. Then $\pi^{-1}\left(\{t\}\right)$ is non-empty; for otherwise, we would have $h(X\cup \{t\})=h(X)$ for every $X$. But it follows from Theorem \ref{thm:structure} that there exists $a\ne t$ such that $h(\{a,t\})\ne h(\{a\})$. Indeed, $a$ can be any element different from $t$ if the arity of $h$ is at most 6 or, if the arity of $h$ is at least $7$, we can take $a$ to be any element different from $t$ belonging to a minimal 2-set of a pure saturation of $h$.


Consider a  minor chain $\left(f_1,\pi_{12},f_2,\pi_{23},f_3, \pi_{34},f_4\right)$ in $\Pol\left(\LO_2,\LO_3\right)$. 
If (at least) two of $f_1,f_2,f_3,f_4$, say $f_i$ and $f_j$ with $i<j$, have small 2-sets, then $\pi_{ij}^{-1}(I(f_j))$ is a 2-set of $f_i$, and hence $\pi_{ij}^{-1}(I(f_j))\cap I(f_i)\ne \emptyset$ by Lemma \ref{lem:2set}, as required.

Assume now that no more than one polymorphism in the chain has a small $2$-set. Suppose $f_i,f_j,f_k$ with $1\leq i<j<k\leq 4$ have no small $2$-set.  
Let the arities of the operations be $n_i,n_j,$ and $n_k$, respectively, and let $t_i,t_j,t_k$ be the dictating variables contained in the choice sets $I\left(f_i\right), I\left(f_j\right)$, and $I\left(f_k\right)$, respectively. 

First we deal with the case when either $n_i\le 6$ or $n_i\ge 7$ and $f_i$ does not have a unique pure saturation. We show that we have $\pi_{ij}(t_i)=t_j$, which implies the desired result.
Assume, for contradiction, that $\pi_{ij}(t_i)=s\ne t_j$. We claim that there exist two disjoint subsets $S'$ and $T'$ of $[n_j]$ such that $s\in S'$, $f_j(S')=0$ and  $t_j\in T'$, $f_j(T')=1$. Indeed, if $n_j\le 6$ then we can take $S'=\{s\}$, $T'=\{t_j\}$, and this will work because $f_j$ is a recoloured projection with dictating variable $t_j$. If $n_j\ge 7$ then one can take any pure saturation $g_j$ of $f_j$ and choose distinct $a,b\in T_{g_j}\backslash \{s,t_j\}$ (this is possible because $\left|T_{g_j}\right|\geq 4$). Then by choosing $S'=\{s,a\}$ and $T'=\{t_j,b\}$, it is easy to see that these have the required properties by using Theorem \ref{thm:structure}. Now let $S=\pi^{-1}_{ij}\left(S'\right)=\{x\in [n_i]\mid \pi_{ij}(x)\in S'\}$ and similarly $T=\pi^{-1}_{ij}\left(T'\right)=\{x\in [n_i]\mid \pi_{ij}(x)\in T'\}$.
Observe that $S$ and $T$ are disjoint (because $S'$ and $T'$ are disjoint) and non-empty (because $t_i\in S$ and $\pi^{-1}_{ij}(t_j)\subseteq T$). Moreover, $f_i(S)=f_j(S')=0\ne [t_i\in S]$ and $f_i(T)=f_j(T')=1\ne [t_i\in T]$. 
 If $n_i\le 6$ then this cannot happen because $f_i$ is a recoloured projection with dictating variable $t_i$, and if $n_i\ge 7$ then this contradicts Lemma \ref{lem:non-unique}. 

Clearly, if either $n_j\le 6$ or $n_j\ge 7$ and $f_j$ does not have a unique pure saturation, the same reasoning as above applies to $f_j$ and $f_k$.

It remains to consider the case when both $f_i$ and $f_j$ have arity at least 7 and both have a unique pure saturation. Let $g_j$ denote the unique saturation of $f_j$.
If $g_j$ has no minimal $2$-set with $n_j-3$ elements, then by Lemma \ref{lem:saturation_commutes_sometimes}, we have that $\pi_{ij}\left(t_i\right)=t_j$ and we are done. If $g_j$ has a minimal $2$-set $T$ with $n_j-3$ elements, since $g_j$ is saturated, this implies that for any $x\in T$, the set $\left(\left[n_j\right]\backslash T\right) \cup \left\{x\right\}$ is a minimal $2$-set of $g_j$ as it has exactly $4$ elements. We therefore have that $T_{g_j}=\left[n_j\right]$. 
Then $\pi_{jk}^{-1}\left(I\left(f_k\right)\right)$ is a (non-empty) $1$-set of $f_j$, and, obviously, $\pi_{jk}^{-1}\left(I\left(f_k\right)\right)\cap T_{g_j}\neq \emptyset$. Lemma \ref{lem:hitting_set_projections} implies that $\pi_{jk}^{-1}\left(I\left(f_k\right)\right)\cap I\left(f_j\right)$ is non-empty, as required.
\end{proof}

We now give an example showing that setting $I(f)$ to be $\{t\}$ for each polymorphism $f$ without small 2-sets would create an issue.

\begin{exmp}\label{exmp:1}
    Consider the following function $g$ of arity 9:
$$
g(S)=
\begin{cases}
			2, & \text{if $|S\cap \{3,\ldots,9\}|\ge 4$}\\
            1, & \text{if $1\le |S\cap \{3,\ldots,9\}|\le 3$ and $1\in S$}\\
            0, & \text{otherwise}
		 \end{cases}
$$
Let $f$ be the same as $g$, but $f(\{3,\ldots,9\})=f(\{1,\ldots,9\})=1$.
One can check that both $f$ and $g$ are polymorphisms of $\left(\LO_2,\LO_3\right)$, and $g$ is the only pure saturation of $f$.

Clearly, the first variable is dictating for $g$, but choosing 
$I(f)=\{1\}$ would lead to a problem. Let $f'=f^\pi$, where $\pi:\left[9\right]\rightarrow \left[2\right]$ is such that $\pi\left(i\right)=1$ for $i\in\left\{1,2\right\}$ and $\pi\left(i\right)=2$ for $i\in\left\{3,\hdots,9\right\}$. It is easy to check that $f'$ is the projection on the second coordinate, so  $I(f')=\{2\}$. But then $\pi^{-1}(I(f'))\cap I(f)=\emptyset$. Moreover, it is easy to see that one can select another minor $f''$ of $f$ such that $f,f'',f'$ is a chain of minors and $f''$ has a small 2-set. 
\end{exmp}


The hardness of $\PCSP\left(\LO_3,\LO_4\right)$ follows easily from Theorem \ref{thm:main}. For those familiar with pp-definitions, the reduction from $\PCSP\left(\LO_2,\LO_3\right)$ follows because the formula $\exists u_x,u_y,u_z R(x,y,z) \wedge R(x,x,u_x) \wedge R(y,y,u_y) \wedge R(z,z,u_z)$ pp-defines $\left(\LO_2,\LO_3\right)$ from $\left(\LO_3,\LO_4\right)$. For others, we spell this reduction out in elementary terms. 

\begin{corollary}[Theorem 1 in \cite{fnotw2024}]
    $\PCSP\left(\LO_3,\LO_4\right)$ is \NP-hard.
\end{corollary}
\begin{proof}
    We give a simple (gadget) reduction from $\PCSP\left(\LO_2,\LO_3\right)$ to $\PCSP\left(\LO_3,\LO_4\right)$. Given an instance $\X$ of $\PCSP\left(\LO_2,\LO_3\right)$, modify it by adding a new element $n_x$ for every element $x$ of $\X$ and adding the tuple $(x,x,n_x)$ to the relation of $\X$. Denote the obtained instance by $\X'$. If $h:\X\rightarrow \LO_2$ then there is a homomorphism $h':\X'\rightarrow \LO_3$ defined as $h'(x)=h(x)$ for all old elements and $h'(n_x)=h(x)+1$ for all new elements. It is easy to see that $h'$ is a homomorphism. 
    Assume that we can efficiently find a homomorphism $h'':\X'\rightarrow \LO_4$. Observe that all the old elements (from $\X$) must be mapped to $\{0,1,2\}$ because of the constraints added when building $\X'$. Thus the restriction of $h''$ to $\X$ is a homomorphism from $\X$ to $\LO_3$.
\end{proof}

\subsection{Proof of Theorem \ref{thm:structure}}\label{sec:proof_of_structure}

The proof of Theorem \ref{thm:structure} can be broken down into the following steps. We warm up by proving it for polymorphisms of arity at most $6$ (Lemma \ref{lem:small_arity}). Next, we treat the case where $\left[n\right]$ can be written as the union of $4$-element $2$-sets (Proposition \ref{prop:minimal_counterexamples} and Lemma \ref{lem:structure_4_element}). We use this as a basis case for an induction argument to show that Theorem \ref{thm:structure} holds for saturated polymorphisms (Lemma \ref{lem:induction_step}). As it turns out, this implies the statement for all polymorphisms of arity $7$ and above thus completing the proof.

\begin{lemma}\label{lem:small_arity}
    Let $f$ be polymorphism of $\left(\LO_2,\LO_3\right)$ of arity $n\leq 6$. Then either 
    $f$ has a small 2-set or there is a unique $t\in [n]$ such that $f$ is a recoloured projection with dictating variable $t$.
\end{lemma}
\begin{proof}
    Pick a polymorphism $f$ of arity $n\leq 6$, and suppose that all of its $2$-sets have $4$ or more elements. We show that these polymorphisms are recoloured projections. We may assume that $4\leq n\leq 6$ for if $n\leq 3$ then $f$ has no 2-sets at all, and so is a polymorphism of $\LO_2$, hence a projection. We therefore have a boolean partition $\left(X,Y,Z\right)$ of $\left[n\right]$ where all parts are non-empty and have at most $2$ elements, so $f$ has a $1$-set of size $1$ or $2$.
    
    Suppose that $X=\left\{1,2\right\}$ is a $1$-set and that every singleton is a $0$-set. This means that there cannot be disjoint $1$-sets, for otherwise, there would be a $2$-set of size $2$ or less. Considering the partition $\left(\left\{1,3\right\},\left\{2\right\},\left[n\right]\backslash\left\{1,2,3\right\}\right)$, we note that $\left\{1,3\right\}$ must be a $1$-set (because $\left[n\right]\backslash \{1,2,3\}$ is disjoint from $\{1,2\}$ and so can only be a 0-set). Now consider any boolean partition $\left(\left\{1\right\},A,B\right)$ where $2\in A$ and $3\in B$. Each of $A$ and $B$ is disjoint from some $1$-set so they both must be $0$-sets. But $\left\{1\right\}$ is then a $1$-set - a contradiction to our assumption on singletons. Hence, we must have a singleton $1$-set $\left\{i\right\}$ for some $i\in [n]$. Every boolean set $S$ which has at least two elements, but does not contain $\{i\}$,  takes part in a boolean partition of $\left[n\right]$ involving $\{i\}$, and must therefore be a $0$-set. Assume that there is another singleton 1-set $\{i'\}$. Then every boolean set $S$ which has at least two elements, but does not contain $i'$, is a 0-set. We cannot have $n\le 5$, as otherwise $[n]\setminus \{i,i'\}$ would be a small 2-set. Then $n=6$ and we can have a partition of $[n]$ into three 2-element sets, each of which does not contain at least one of $i$ and $i'$.
    All sets in this partition would be 0-sets, which is impossible. Finally, every boolean set $S$ with $i\in S$ that is part of a boolean partition must be a $1$-set, since the boolean sets participating in a partition with $S$ must be $0$-sets. From this, it follows that $f$ has the required form. 
\end{proof}

From now on, assume that the arity $n$ of $f$ is at least 7.

We now look at polymorphisms $f$ for which $\left[n\right]$ is the union of $2$-sets of size $4$. That these have the structure described in Theorem \ref{thm:structure} will follow easily from the following statement.
\begin{proposition}\label{prop:minimal_counterexamples}
    Let $f$ be a $n$-ary polymorphism of $\left(\LO_2,\LO_3\right)$ such that $\left[n\right]$ is the union of $4$-element $2$-sets. If $f$ has no small $2$-sets then:
    \begin{enumerate}
        \item $f$ has no disjoint static $1$-sets.
        \item $f$ has a singleton static $1$-set.
    \end{enumerate}
\end{proposition}
We prove both parts of the above proposition by contradiction, via a minimal counterexample reasoning. 

To show Proposition \ref{prop:minimal_counterexamples}(1), assume that there exists an $f$ satisfying the assumptions and having disjoint static $1$-sets, and pick one such $f$ satisfying the following extremal conditions: the arity of $f$ has to be smallest possible, and amongst all such polymorphisms of minimal arity we pick $f$ to have largest (in size) disjoint union of two static $1$-sets. 
Fix such static disjoint $1$-sets $X$ and $Y$. Note that $\overline{X\cup Y}$ must be a 2-set of $f$ and so $\left|X\cup Y\right|\le n-4$.

\begin{lemma}\label{lem:saturated_WLOG}
    $f$ can be chosen to be saturated.
\end{lemma}
\begin{proof}
    Observe that for any disjoint $A,B$ such that $A\cup B=\left[n\right]$, we cannot have simultaneously $A\subseteq X$ and $B\subseteq Y$ (or $A\subseteq Y$ and $B\subseteq X$). Also note that in $f$, no $2$-set can be a subset of $X\cup Y$, since $\overline{X\cup Y}$ is a 2-set. We construct a saturated polymorphism $\hat{f}$ from $f$ as follows. First, we change the value of any superset of a $2$-set to $2$ - note that from the previous remark this does not change the values of $X$ and $Y$. Next, for complementary sets $A,B$, if $A\subseteq X$ or $A\subseteq Y$, then $B$  contains $\overline{X\cup Y}$ so it is already a 2-set. If neither $A$ nor $B$ are subsets of $X$ or $Y$, and both of them are boolean, then arbitrarily pick which of $A$ and $B$ to be changed to a $2$-set, so long as the chosen $2$-set has size at least $4$ (this is possible because $n\geq 7$). Apply these rules exhaustively. By Lemmas \ref{lem:upwards_closure} and \ref{lem:complementarity}, we obtain a saturated polymorphism $\hat{f}$. 
    In addition, the sets $X$ and $Y$ are still 1-sets and they retain their maximality property: indeed, by Lemma \ref{lem:boolean_recolouring} neither is $0$-recolourable, since $\hat{f}$ is saturated and they both intersect non-trivially with some $4$-element minimal $2$-set. Finally, observe that all minimal $2$-sets of $\hat{f}$ have size at least $4$, as the construction did not introduce any $2$-sets with fewer than $4$ elements. Moreover, all variables are contained in some $4$-element minimal $2$-set, since the $2$-sets of $f$ retain the same value in $\hat{f}$. We have therefore constructed a saturated polymorphism $\hat{f}$ satisfying the assumptions in Proposition \ref{prop:minimal_counterexamples} with the same extremal properties of $f$ (arity and the size of $X\cup Y$).
\end{proof}
Let us therefore assume that $f$ is saturated. Lemmas \ref{lem:Y_good_intersection}-\ref{lem:split} will lead to a contradiction showing that such an $f$ does not exist.
\begin{lemma}\label{lem:Y_good_intersection}
    For any distinct $a,b\in Y$ there exist a $4$-element minimal $2$-set $T$ such that $a,b\in T$.
    The same holds for any distinct $a,b\in X$.
\end{lemma}

\begin{proof}
    Suppose that for some $a,b\in Y$, there is no $4$-element minimal $2$-set $T$ such that $a,b\in T$. Consider the polymorphism $f^\pi$ where $\pi$ maps $\left\{a,b\right\}$ to $b$ and acts as identity on the rest of $\left[n\right]$. Then $f^\pi$ has no small $2$-sets. In addition, all variables are contained in some $4$-element minimal $2$-set. To see this, let $S$ be a $4$-element minimal $2$-set of $f$. If $a,b\notin S$, it must be that $S$ is still a minimal $2$-set of $f^\pi$. By assumption, we cannot have both $a,b\in S$ so we just need to consider the case where $a\in S$ but $b\notin S$. In this case we have, by saturation of $f$, that $\left(S\backslash\left\{a\right\}\right)\cup\left\{b\right\}$ is a $4$-element minimal $2$-set of $f^\pi$. Thus, all variables of $f^\pi$ are contained in some $4$-element $2$-set. In addition, $X$ and $Y\backslash\left\{a\right\}$ are both disjoint $1$-sets for $f^\pi$. As they both non-trivially intersect a minimal $2$-set and $f^\pi$ is also saturated, we reach a contradiction to the minimality of $n$. The case of $a,b\in X$ is symmetric.
\end{proof}

\begin{lemma}\label{lem:monotonicity}
    Let $T$ be 2-set and $A$ a 1-set such that $S=T\backslash A$ is boolean. For any $a\not\in T\cup A$, if $S\cup\left\{a\right\}$ is boolean then $f(A\cup\left\{a\right\})\ge 1$.
\end{lemma}

\begin{proof}
Let $R=\overline{T\cup A}$. Consider the partitions $(S,A,R),  (S\cup\left\{a\right\},A,R\backslash\left\{a\right\}), (S,A\cup\left\{a\right\},R\backslash\left\{a\right\})$. If $A\cup\left\{a\right\}$ is a 2-set, we are done. Otherwise, these three are all boolean partitions. The first two imply that $S$ and $R\backslash\left\{a\right\}$ are 0-sets, so $A\cup\left\{a\right\}$ is a 1-set.
\end{proof}

\begin{lemma}\label{claim:no-disjoint}
    If $A$ and $B$ are 1-sets such that $|A\cap T_A|\ge 2$ and $|B\cap T_B|\ge 2$ for some 4-element 2-sets $T_A,T_B$ then $A\cap B\ne\emptyset$.
\end{lemma}
\begin{proof}
Assume for contradiction that there are disjoint 1-sets $A$ and $B$ such that $|A\cap T_A|\ge 2$ and $|B\cap T_B|\ge 2$ for some 4-element 2-sets $T_A,T_B$. Choose them so that $A\cup B$ has the largest possible size.
Let $S_A=T_A\backslash A$ and $S_B=T_B\backslash B$. Note that $|S_A|\le 2$ and $|S_B|\le 2$, so in particular 
$S_A\cup\left\{a\right\}$ is boolean for any $a\not\in A\cup B\cup T_A$ and, similarly, $S_B\cup\left\{b\right\}$ is boolean for any $b\not\in A\cup B\cup T_B$.
Then Lemma~\ref{lem:monotonicity} implies that, for any element $a\not\in A\cup B\cup T_A$, we have $f(A\cup\left\{a\right\})\ge 1$.
We cannot have $f(A\cup\left\{a\right\})=1$, since this would contradict the choice of $A,B$. So we have $f(A\cup\left\{a\right\})=2$.
Similarly, we have $f(B\cup\left\{b\right\})=2$ for every $b\not\in A\cup B\cup T_B$. Since $|S_A|\le 2$ and $|\overline{A\cup B}|\ge 4$,
there is more than one choice for $a\not\in A\cup B\cup T_A$. Similarly, there is more than one choice for $b\not\in A\cup B\cup T_B$. But then, choosing $a$ and $b$ to be different, we obtain disjoint 2-sets $A\cup\left\{a\right\}$ and $B\cup\left\{b\right\}$, a contradiction. 
\end{proof}

Lemma \ref{claim:no-disjoint} together with Lemma~\ref{lem:Y_good_intersection} immediately imply the following.

\begin{corollary}
    At least one of $X,Y$ must be a singleton.
\end{corollary}

We may therefore assume that $X$ contains a single element $x$.

\begin{lemma}\label{lem:type}
   If $T$ is a $2$-set that can be written as a disjoint union of two boolean sets, then it is of one of the following two types:
    \begin{enumerate}
        \item Type $01$, if $T$ is the disjoint union of a $1$-set and a $0$-set.
        \item Type $00$, if $T$ is the disjoint union of two $0$-sets.
    \end{enumerate}
    In addition, the above types are mutually exclusive.
\end{lemma}
\begin{proof}
    Let $T$ be as in the statement and suppose $T=A\cup B$ where $A,B$ are disjoint boolean sets. If $\overline{T}$ is a $1$-set, then $A,B$ are both $0$-sets. If, instead, $\overline{T}$ is a $0$-set, then $A,B$ have different values.
\end{proof}

\begin{lemma}\label{lem:split}
    Let $X=\{x\}$ be as above and let $S$ be a $4$-element minimal $2$-set such that $x\in S$. Then, for any $z\not\in S$, $\{z,x\}$ is a $1$-set.
\end{lemma}

\begin{proof}
    Let $Z=\left\{z,x\right\}$ and consider the $2$-set $S\cup Z$. Note that $S$ is of type $01$ (in the sense of Lemma \ref{lem:type}), since $\left\{x\right\}$ is a $1$-set and therefore $S\backslash\left\{x\right\}$ is a $0$-set by Lemma \ref{lem:type}. By partitioning $S$ into two disjoint $2$-element boolean sets, one of these must therefore be a $1$-set which we refer to as $U$. Since $\left|\left(S\cup Z\right)\backslash U\right|=3$, we must have that $\left(S\cup Z\right)\backslash U$ is boolean. Since $U$ is a $1$-set, $S\cup Z$ must be of type $01$ by Lemma \ref{lem:type}. But $f\left(\left(S\cup Z\right)\backslash Z\right)=f\left(S\backslash Z\right)=f\left(S\backslash\left\{x\right\}\right)=0$, so it follows that $Z$ is a $1$-set as required.
\end{proof}

The above lemma is sufficient to show that $f$ does not exist, for it contradicts maximality of $\left|X\cup Y\right|$. Indeed, $\overline{X\cup Y}$ is a $2$-set and it therefore contains at least $4$ elements. Choose a 4-element 2-set $S$ of $f$ such that $x\in S$, which exists because every element of $[n]$ is contained in a 4-element 2-set. Then $\overline{X\cup Y}\setminus S$ is non-empty, and we can choose an element $z$ from this set. By Lemma \ref{lem:split}, $\left\{z,x\right\}$ is a $1$-set which, by saturation, is static. In particular, $\{z,x\}$ and $Y$ are disjoint static 1-sets of $f$,
implying that $\left|X\cup Y\right|$ is not maximal - a contradiciton. Thus, Proposition \ref{prop:minimal_counterexamples}$(1)$ holds. 

For Proposition \ref{prop:minimal_counterexamples}(2), assume ad absurdum that there is a polymorphism $f$ with the same assumptions on $2$-sets but no static $1$-set singletons. Pick $f$ so that the arity is minimal and its smallest $1$-set the largest (amongst all polymorphisms with minimal arity). Using the same argument as Lemma \ref{lem:saturated_WLOG}, one can show that it is possible to recolour $f$ to a saturated polymorphism without affecting the extremal properties (minimal arity and maximal smallest $1$-set), so let us further assume that $f$ is saturated. As before, we prove a string of properties which lead to contradict the extremal assumptions on $f$.

\begin{lemma}\label{lem:2-element}
    For any $a,b\in [n]$, if $\{a,b\}$ is not contained in a 4-element 2-set then $\{a,b\}$ is a 1-set.
\end{lemma}

\begin{proof}
    Follows using an analogous argument to Lemma \ref{lem:Y_good_intersection}.
\end{proof}

\begin{lemma}\label{lem:upwards-monotone}
If $X$ and $X'$ are non-recolourable boolean sets such that $X\subset X'$ then $f(X)\le f(X')$. 
\end{lemma}

\begin{proof}
    Assume, for a contradiction, that $X$ is a 1-set and $X'$ is a 0-set. Then $X'$ is a part of some boolean partition of $\left[n\right]$. So there exists a 1-set $Y$ such that $X'\cap Y=\emptyset$. In particular, $X$ and $Y$ are disjoint static 1-sets - a contradiction to Proposition \ref{prop:minimal_counterexamples}$\left(1\right)$.
\end{proof}

\begin{lemma}\label{lem:ABC}
    There exist three pairwise disjoint sets $A,B,C$, with at most two elements each, with the following property: any $X$ such that $X\cap (A\cup B\cup C)$ is one of $A,B,C$ is not a 1-set.
\end{lemma}

\begin{proof}
There is a 1-set with 2 or 3 elements by Lemma \ref{lem:kneser}. Assume first that we have such a set that is contained in a 4-element 2-set $T=\{a,b,c,d\}$.
Then $T$ is of type 01 in the sense of Lemma \ref{lem:type}, and hence some 2-element subset, say $\{a,b\}$, of $T$ is a 1-set. Since $\{a\}$ and $\{b\}$ are not 1-sets, it follows that both $\{a,c,d\}$ and $\{b,c,d\}$ are 1-sets. We take $A=\{a\}$, $B=\{b\}$, and $C=\{c,d\}$ and show that they satisfy the required property.
Any set $X$ such that $X\cap (A\cup B\cup C)$ is one of $A,B,C$ is disjoint from   one of
$\{a,b\}$, $\{a,c,d\}$ and $\{b,c,d\}$ and so cannot be 1-set.

We now prove that some 1-set with 2 or 3 elements must be contained in a 4-element 2-set. Assume, for contradiction, that this is not the case. We start by showing that there must exist a minimal 1-set with 2 elements. 
If $\{a,b,c\}$ is a minimal 1-set, then none of its 2-element subsets is a 1-set. So by Lemma \ref{lem:2-element}, $\{a,b\}$ is contained in a 4-element 2-set, say $T=\{a,b,d,e\}$. Then $T\cup \{c\}$ is 2-set of type 01. The set $\{a,b,e\}$ is contained in a 4-element 2-set, so it cannot be a 1-set by our assumption. Hence $\{c,d\}$ is 1-set.
Assume now that the 1-set $\{c,d\}$ is not contained in any 4-element 2-set. By assumption on $f$, each of $c$, $d$ is contained in a 4-element 2-set. Denote these as $T_c=\{c,c_1,c_2,c_3\}$ and $T_d=\{d,d_1,d_2,d_3\}$. Consider the 2-set $T_c\cup \{d\}$. It is of type 01 since $\left\{c,d\right\}$ is a $1$-set. If some proper subset of $T_c$ is a 1-set then we are done. So we can assume that they are all 0-sets, so each of $\{d,c_1\}, \{d,c_2\}, \{d,c_3\}$ is a 1-set.
Similarly, if none of the subsets of $T_d$ is a 1-set, we have have that each of $\{c,d_1\}, \{c,d_2\}, \{c,d_3\}$ is a 1-set. Since 1-sets cannot be disjoint, we have that $\{c,d_1\}$ intersects each of $\{d,c_1\}, \{d,c_2\}, \{d,c_3\}$. But this is impossible because $c\ne d\ne d_1$ and $c_1,c_2,c_3$ are all distinct.
\end{proof}

\begin{lemma}\label{lem:existence}
 $f$ does not exist.
\end{lemma}

\begin{proof}
    Take sets $A,B,C$ guaranteed by Lemma \ref{lem:ABC}. Denote $U=A\cup B\cup C$. We will look at partitions $(X_a,X_b,X_c)$ of $[n]$
where $A\subseteq X_a$, $B\subseteq X_b$, $C\subseteq X_c$. The sets involved in such partitions cannot be 1-sets. Consider the function $f'$ defined on $2^{[n]\backslash U}$, which on input $X$ outputs the triple $(f(X\cup A),f(X\cup B),f(X\cup C))\in \{0,2\}^3$.

We claim that either this triple is the same for all $X$ or it is constant (i.e. contains three equal entries) for every $X$.
Put otherwise, either this triple does not depend on $X$ or the values in it do not depend on $A,B,C$ (but may depend on $X$). The reader may wish to use Table \ref{fig:table} as an aid. Assume that, for some $X\subseteq 2^{[n]\setminus U}$, the triple $(f(X\cup A),f(X\cup B),f(X\cup C))$ is not constant. Consider any partition $(X,Y,Z)$ of $[n]\setminus U$. Observe that if $f(X\cup A)=2$ then none of $Y\cup B,Y\cup C,Z\cup B,Z\cup C$ can be a 2-set, so they are all 0-sets. If $f(X\cup B)=0$ then neither $Y\cup A$ nor $Z\cup A$ can be a 0-set, so they are 2-sets. Finally, $X\cup C$ cannot be a 2-set, so the triples $(f(X\cup A),f(X\cup B),f(X\cup C))$, $(f(Y\cup A),f(Y\cup B),f(Y\cup C))$, $(f(Z\cup A),f(Z\cup B),f(Z\cup C))$ are all equal. For any subset $X\subseteq [n]\setminus U$, one can consider the partition $(X,\overline{X},\emptyset)$, and therefore the triple
$(f(X\cup A),f(X\cup B),f(X\cup C))$ is the same as $(f(\emptyset\cup A),f(\emptyset\cup B),f(\emptyset\cup C))$, so indeed the triple $(f(X\cup A),f(X\cup B),f(X\cup C))$ is the same for all $X$.

If it is the case that $(f(X\cup A),f(X\cup B),f(X\cup C))$ is not constant for some $X$, then this is true for all $X$. But picking $X$ to be the empty set, we get that one of $A,B,C$ is a 2-set, a contradiction.

So assume now that $(f(X\cup A),f(X\cup B),f(X\cup C))$ is a constant in $\{0,2\}$ for every $X$, so we can identify the value $f'(X)$ with this constant. Now let $f''$ be the function on $\{0,1\}$ obtained from $f'$ by replacing all values 2 with value 1. Then $f''$ is a polymorphism of $\LO_2$, and it is well known that every such polymorphism is a projection. Then there is some singleton $\{x\}$ such that $f''(\{x\})=1$, which means that $\{x\}\cup A$, $\{x\}\cup B$, $\{x\}\cup C$ are all 2-sets for $f$ - a contradiction, since each of $A,B,C$ has at most $2$ distinct elements.
\end{proof}

This concludes the proof of Proposition \ref{prop:minimal_counterexamples}(2).

\begin{table}[]
    \centering
    \begin{tabular}{c||c|c|c}
         & $X$ & $Y$ & $Z$\\
         \hhline{=||=|=|=}
       $A$  & $f\left(X\cup A\right)$ & $f\left(Y\cup A\right)$  & $f\left(Z\cup A\right)$ \\
       \hline
       $B$ & $f\left(X\cup B\right)$ & $f\left(Y\cup B\right)$ & $f\left(Z\cup B\right)$\\
       \hline
       $C$ & $f\left(X\cup C\right)$ & $f\left(Y\cup C\right)$ & $f\left(Z\cup C\right)$
    \end{tabular}
    \caption{The proof of Lemma \ref{lem:existence} shows that either the rows of the table are constant or the columns are constant. It is useful to note that any three entries that are pairwise not in the same row or column (for example, $f\left(X\cup C\right)$, $f\left(Y\cup A\right)$, and $f\left(Z\cup B\right)$)  must form a triple in $LO_3$.}
    \label{fig:table}
\end{table}

\begin{lemma}\label{lem:structure_4_element}
    Theorem \ref{thm:structure} holds when $\left[n\right]$ is the union of $4$-element $2$-sets.
\end{lemma}
\begin{proof}
    Let $f$ be a polymorphism satisfying the conditions in the statement and suppose it has no small $2$-set. By Proposition \ref{prop:minimal_counterexamples}(2), $f$ must have a singleton static $1$-set $\left\{x\right\}$. It follows from Proposition \ref{prop:minimal_counterexamples}(1) that every boolean static set not containing $x$ is a 0-set. Suppose by contradiction that there is some static $0$-set $X$ such that $x\in X$. Then there must be a boolean partition $\left(X,Y,Z\right)$. But then one of the sets $Y,Z$ must be a static $1$-set, which means that we have two disjoint $1$-sets that are not $0$-recolourable - a contradiction to Proposition \ref{prop:minimal_counterexamples}(1). Hence, any pure saturation of a polymorphism $f$ as in the statement of Theorem \ref{thm:structure} has the predicted structure.
\end{proof}

For the next case of Theorem \ref{thm:structure}, we use Lemma \ref{lem:structure_4_element} as a basis for an induction argument.

\begin{lemma}\label{lem:induction_step}
    Theorem \ref{thm:structure} holds for saturated polymorphisms.
\end{lemma}

\begin{proof}
    We proceed by induction on the number of variables not contained in a $4$-element minimal $2$-set. If there are no such variables, the statement holds by Lemma \ref{lem:structure_4_element}.

    Suppose that, for some $m\geq 0$, the statement holds for polymorphisms with at most $m$ variables not contained in a $4$-element minimal $2$-set. Consider a saturated polymorphism $f$ of arity $n$ and suppose it has exactly $m+1$ variables not contained in a $4$-element minimal $2$-set. If $f$ has a small $2$-set then we are done. Otherwise, assume without loss of generality that $n$ is a variable that is not contained in any $4$-element minimal $2$-set.
    For each $i\in\left[n-1\right]$, define $\pi_i:\left[n\right]\rightarrow\left[n-1\right]$ as $\pi_i\left(n\right)=i$ and $\pi_i\left(j\right)=j$ for all $j\neq n$, and set $g_i=f^{\pi_i}$. Note that each $g_i$ is saturated and has at most $m$ variables not contained in any $4$-element minimal $2$-set. The former property easily follows from the assumption that $f$ is saturated. The latter property follows from the fact that every 4-element 2-set $T$ of $f$ is (a subset of $[n-1]$ and) a 2-set of $g_i$. Indeed, the set $\pi_i^{-1}(T)$ contains $T$ and so is a 2-set of $f$, and so $g_i(T)=f(\pi_i^{-1}(T))=2$.  By the induction hypothesis there exist, for each $i\in\left[n-1\right]$ some element $t_i\in [n-1]$ such that, for every static boolean set $S$ of $g_i$, we have $g_i\left(S\right)=1$ if, and only if, $t_i\in S$. Observe that, for each $i\in\left[n-1\right]$, $j\in T_f$ and $k\in\left[n-1\right]\backslash \left\{t_i,j\right\}$ it holds that
    \begin{equation}\label{eq:induction}
        f(\{t_i,j\})=g_k(\{t_i,j\})=[t_k\in \{t_i,j\}],
    \end{equation}
    where $\left[\star\right]$ is the Iverson bracket.

    We now assume that $t_i\neq i$ for some $i$ and show that in this case $t_k$ is the same for all $k$. For this $i$ and any $j\in T_f\backslash \{i\}$, we have  $f(\{t_i,j\})=g_i(\{t_i,j\})=[t_i\in \{t_i,j\}]=1$. Comparing this with Equation (\ref{eq:induction}), we get that, $t_k\in \{t_i,j\}$, for all $k\in [n-1]\backslash \{t_i,j\}$. If $t_k\ne t_i$ for some $k\ne t_i$, then we must have $t_k=j$ for any choice of $j\in T_f\backslash \{i,k,n\}$. This is impossible, since $T_f$ has at least 7 elements by saturation, so $t_k=t_i$ whenever $k\ne t_i$. If $t_k\ne t_i$ for $k=t_i$ then we can repeat the above reasoning starting with $k$ instead of $i$ and derive that all elements $t_s$ except possibly for one are equal to $t_k$. But we already have this conclusion for $t_i$, which contradicts the assumption $t_k\ne t_i$. We therefore conclude that $t_i$ is independent of $i$. 
    
    We proved that either $t_i=i$ for all $i\in [n-1]$ or else all $t_i$ are the same. We now set $t$ to be $n$ in the former case and $t_i$ in the latter case, and show that $f$ has the required structure with this choice of $t$. Let $S_0$ be a static boolean set for $f$ and let $(S_0,S_1,S_2)$ be a boolean partition of $[n]$ witnessing this. Assume first that none of the sets $S_0,S_1,S_2$ is $\{n\}$. Then, for any $i\in [n-1]$ from the same part of the partition $(S_0,S_1,S_2)$ as $n$, we have that $(\pi_i(S_0),\pi_i(S_1),\pi_i(S_2))$ is a boolean partition of $[n-1]$ for $g_i$, and hence $f(S_j)=g_i(\pi_i(S_j))=[t_i\in \pi_i(S_j)]=[t\in S_j]$ for $j=0,1,2$ (note that we have $[t_i\in \pi_i(S_j)]=[t\in S_j]$ for both settings of $t$). In particular, $f(S_0)=[t\in S_0]$, as required. If $S_1=\{n\}$, then $S_2\ne\emptyset$, for otherwise $S_1\cup S_2$ would be boolean and $S_0$ a 2-set by saturation. Then, for any $a\in S_2$, the partition $(S_0,S_1\cup \{a\},S_2\backslash \{a\})$ is boolean and none of its parts is $\{n\}$, which we already showed to imply $f(S_0)=[t\in S_0]$. The reasoning for the case $S_2=\{n\}$ is the same.
    It remains to consider the case when $S_0=\{n\}$. The above reasoning (with an appropriate renaming of the sets $S_0,S_1,S_2$) shows that we have $f(S_j)=[t\in S_j]$ for $j=1,2$. Thus $f(\{n\})=[t\in \{n\}]$, as required. 
\end{proof}

We finally show that Lemma \ref{lem:small_arity} and Lemma \ref{lem:induction_step} imply Theorem \ref{thm:structure}.

\begin{proof}[Proof of Theorem \ref{thm:structure}.]
    Lemma \ref{lem:small_arity} takes care of polymorphisms of arity $n\leq 6$. For polymorphisms with arity $n\geq 7$, Lemma \ref{lem:induction_step} takes care of the saturated case. Suppose $f$ has arity $n\geq 7$ and all its $2$-sets have $4$ or more elements. We show that if $g$ and $h$ are pure saturations of $f$, then $g$ and $h$ have the same dictating variable. Indeed, suppose $g,h$ have dictating variables $t,t'$ respectively with $t\neq t'$. Let $T_g,T_h$ be the unions of minimal $2$-sets of $g$ and $h$ respectively. Pick $i\in T_g,j\in T_h$ with $t'\notin\left\{i,j\right\}$ and consider $S=\left\{t,i,j\right\}$. For both $g$ and $h$, the set $S$ is boolean (because it has at most three elements) and static by Lemma \ref{lem:boolean_recolouring}. Moreover, we have $f\left(S\right)=g\left(S\right)=h\left(S\right)$, since $S$ is not affected by any saturation path from $f$ to $g$ or $f$ to $h$. But by Lemma \ref{lem:induction_step} we have that $g\left(S\right)=\left[t\in S\right]=1$ and $h\left(S\right)=\left[t'\in S\right]=0$ - a contradiction.
\end{proof}

    \bibliographystyle{plain}
\bibliography{main.bbl}
\end{document}